\documentclass[12pt,oneside,reqno]{amsart}

\usepackage{braket}
\usepackage{slashed}
\usepackage{mathtools}
\usepackage{graphicx}
\usepackage{amsmath}
\usepackage{marginnote}
\usepackage{xcolor}
\usepackage{multicol}
\usepackage{hyperref}
\usepackage[margin=1.25in]{geometry}
\usepackage{tikz,xcolor}

\definecolor{lime}{HTML}{A6CE39}
\DeclareRobustCommand{\orcidicon}{
	\begin{tikzpicture}
	\draw[lime, fill=lime] (0,0) 
	circle [radius=0.16] 
	node[white] {{\fontfamily{qag}\selectfont \tiny ID}};
	\draw[white, fill=white] (-0.0625,0.095) 
	circle [radius=0.007];
	\end{tikzpicture}
	\hspace{-2mm}
}

\foreach \x in {A, ..., Z}{\expandafter\xdef\csname orcid\x\endcsname{\noexpand\href{https://orcid.org/\csname orcidauthor\x\endcsname}
			{\noexpand\orcidicon}}
}

\hypersetup{
    colorlinks=true,
    linkcolor=blue,
    filecolor=magenta,      
    urlcolor=cyan,
    citecolor=red,
    pdftitle={separabilitypaper},
    pdfpagemode=FullScreen,
    }
\graphicspath{ {./images/} }
\newtheorem{theorem}{Theorem}[section]
\newtheorem{lemma}[theorem]{Lemma}
\newtheorem{proposition}[theorem]{Proposition}

\theoremstyle{definition}
\newtheorem{definition}[theorem]{Definition}
\newtheorem{example}[theorem]{Example}

\theoremstyle{remark}

\DeclareMathOperator{\perm}{perm}
\DeclareMathOperator{\tr}{Tr}

\numberwithin{equation}{section}

\begin{document}

\title[Cycle Index Polynomials and Quantum Separability Tests]{Cycle Index Polynomials and Generalized Quantum Separability Tests}
\author{$^1$Zachary P. Bradshaw\orcidA{}}
\address{$^1$Tulane University, Department of Mathematics}
\email{zbradshaw@tulane.edu}
\author{$^2$Margarite L. LaBorde\orcidB{}}
\address{$^2$Louisiana State University, Hearne Institute for Theoretical Physics and Department of Physics \& Astronomy}
\email{mlabo15@lsu.edu}
\author{$^{2,3}$Mark M. Wilde\orcidC{}}
\address{$^3$Cornell University, School of Electrical and Computer Engineering}
\email{wilde@cornell.edu}
\maketitle
\begin{abstract}
	The mixedness of one share of a pure bipartite state determines whether the overall state is a separable, unentangled state. Here we consider quantum computational tests of mixedness, and we derive an exact expression of the acceptance probability of such tests as the number of copies of the state becomes larger. We prove that the analytical form of this expression is given by the cycle index polynomial of the symmetric group $S_k$, which is itself related to the Bell polynomials. After doing so, we derive a family of quantum separability tests, each of which is generated by a finite group; for all such algorithms, we show that the acceptance probability is determined by the cycle index polynomial of the group. Finally, we produce and analyze explicit circuit constructions for these tests, showing that the tests corresponding to the symmetric and cyclic groups can be executed with $O(k^2)$ and $O(k\log(k))$ controlled-SWAP gates, respectively, where $k$ is the number of copies of the state being tested.
 \end{abstract}

\section{Introduction}\label{introduction}

Entanglement is an inherently quantum mechanical phenomenon and thus is of great interest in quantum information science. Separability tests acting on pure states equivalently test for entanglement, thus creating an avenue to identify potentially resourceful states and verify entanglement generation protocols. Much work has been done to understand entanglement, including the positive partial transpose (PPT) criterion \cite{Per96,HHH96} and the related entanglement measure known as logarithmic negativity \cite{ZHSL98,EP99,Vidal2002}. 
Another well known approach considers $k$-extendibility  of a state \cite{W89a,DPS02,navascues2009}, with related quantifiers proposed in \cite{KDWW19,WWW19,KDWW21}. Experimental entanglement tests have also been implemented on trapped ion \cite{brydges2019} and optical systems \cite{shalm05,cassemiro2010}.  These notions exemplify how entanglement is characterized and quantified in the current literature.

The most common quantum computational test of separability of pure states is the swap test, introduced in \cite{barenco1997stabilization}, applied in \cite{brennen2003entanglement} for entanglement detection, and used in quantum fingerprinting \cite{buhrman2001quantum}. To understand it, first recall that a pure bipartite state $\ket{\psi}_{AB}$ is separable if it can be written as a tensor product of two states, as
\begin{align}
    \ket{\psi}_{AB} =\ket{\phi}_A\otimes\ket{\varphi}_B.
    \label{eq:sep-pure-state}
\end{align}
Now, if we take two copies of this separable state, it has the following form:
\begin{align}
    \ket{\psi}_{A_1 B_1} \otimes \ket{\psi}_{A_2 B_2} =\ket{\phi}_{A_1} \otimes \ket{\varphi}_{B_1} \otimes \ket{\phi}_{A_2} \otimes \ket{\varphi}_{B_2}.
\end{align}
This state is invariant under a swap of systems $A_1$ and $A_2$, as well as a swap of systems $B_1$ and $B_2$. Thus, the swap test accepts with certainty in this case; however, if a pure bipartite state is not separable, the two-copy state does not possess the above swap invariance, and the swap test can detect this lack of invariance by means of the  phase kickback trick, well known in quantum computation. Here we make use of the swap test, as well as generalizations of it, to study the symmetry of quantum states with respect to an arbitrary finite group \cite{burchardt2021,margo}.

The recent work in \cite{margo} proposed a generalization of the swap test as a method for detecting entanglement, based on the observation that multiple copies of the separable state in \eqref{eq:sep-pure-state} are invariant under arbitrary permutations of both the $A$ systems and $B$ systems. Indeed, by writing such a state down explicitly as
\begin{equation}
    \bigotimes_{i=1}^{k} \ket{\psi}_{A_i B_i} = \bigotimes_{i=1}^{k} \ket{\phi}_{A_i} \otimes \ket{\varphi}_{B_i},
\end{equation}
it is clear that such a state is invariant as mentioned above. However, if the state $\ket{\psi}_{A B}$ is not separable, then checking for various kinds of permutation invariance of the state $\bigotimes_{i=1}^{k} \ket{\psi}_{A_i B_i}$ leads to more fine-grained tests of entanglement with alternative mathematical expressions for the acceptance probability of the test. Other works, such as \cite{islam2015measuring,nguyen2021experimental} express interest in systems of a similar nature. Understanding these expressions in more detail is one of the main goals of the present paper. Our intention is to generalize the family of separability tests in \cite{margo} further by establishing a one-to-one correspondence between finite groups and a subset of quantum separability tests and to compare the cost of implementing these new tests with that of the swap and symmetric groups tests.

Before we discuss our results, let us indicate how our approach falls within the framework of $G$-Bose symmetry tests, as developed in \cite{margo}. We use this method to discuss the separability of a pure bipartite state \cite{DPS02}, although group symmetry tests have been used to great effect to test other properties as well (see, e.g. \cite{kada2008efficiency}). These symmetry tests involve a group~$G$ and determine if a given state is invariant under the action of a unitary representation of that group. Bose symmetry tests specifically describe situations where every element of a group is considered simultaneously rather than consecutively. The quantum algorithm is implemented by physically realizing the projector onto the symmetric subspace, and the symmetry of the state is manifested through the acceptance probability of the algorithm (see \cite{harrow2013church} for a thorough discussion of the symmetric subspace in regards to quantum information). In the $G$-Bose symmetry framework, a pure state $\ket{\psi}_{AB}$ is given and an $S_k$-Bose symmetry test is conducted on the tensor-power state $\ket{\psi}_{AB}^{\otimes k}$, where $S_k$ denotes the symmetric group on $k$ letters. Possessing this tensor-power state is equivalent to having access to $k$ copies of our state being tested. The swap test is recovered as a special case in which $k=2$.

The $G$-Bose symmetry tests allow for a generalization of the swap test to more copies of a state of interest and higher-order groups. These algorithms exchange simplicity for certainty, analogous to how fingerprinting is both more accurate and complicated when greater numbers of prints are taken. In choosing to investigate group symmetries, rather than merely the swap test, the separability of a state can be determined more quickly and accurately. In what follows, we will explicitly consider the trade-off between the swap test, which may pass some states with high probability even if entangled, and higher-order permutation tests, which will reject entangled states with higher probability than a swap test in exchange for the use of more resources.

The natural question is, when do these more complex tests merit performing? In this work, we derive the acceptance probability of a generalized separability test. In doing so, we present an inherent reliance on the cycle index polynomial, a particularly important polynomial in Pólya theory \cite{polya,combinatorics} that encodes the structure of a permutation group by storing the number of elements of a given cycle type as its coefficients. This allows us to compare separability tests generated from various groups, as well as investigate the mathematical relationships inherently present in these tests. We directly show that an arbitrary finite group generates a separability test with its acceptance probability given by the cycle index polynomial of that group. We supplement this by then giving explicit quantum circuit descriptions for groups of interest and counting the number of gates needed to realize each test. Combining our acceptance probability results with resource counting gives us a metric to compare when the relative strictness of the test is outweighed by the benefit of fewer gate resources, and we discuss this factor in more detail in Section~\ref{sec:comparison}.

The rest of our paper is organized as follows. In Section~\ref{formula}, we review the algorithm for the  bipartite pure-state separability test in \cite{margo}, and we prove  that the acceptance probability of this algorithm is given by the cycle index polynomial \cite{polya,combinatorics} of the symmetric group $S_k$, which is itself related to the complete Bell polynomials \cite{Roman}. This enables us to write the acceptance probability as both the permanent and the determinant of particular matrices, as a consequence of Newton's identities \cite{Newton}. In~\cite{margo}, it was conjectured that this acceptance probability does not increase as $k\to\infty$. In Section~\ref{sec:conjecture}, we prove that this conjecture is true. In fact, we show that it strictly decreases and converges to zero whenever $\rho_B \coloneqq \tr_A[|\psi\rangle\!\langle\psi |_{AB}]$ is not a pure state.

In Section~\ref{generalization}, we generalize the  bipartite pure-state separability test to an algorithm involving any group $G$, in which a $G$-Bose symmetry test is performed on the tensor-power state $\ket{\psi}_{AB}^{\otimes k}$. By identifying $G$ with a subgroup of $S_k$, which is guaranteed to exist by Cayley's theorem, we  show, by the same reasoning as in Section~\ref{formula}, that the acceptance probability of the algorithm is given by the cycle index polynomial of the group $G$. We discuss how these generalized tests are in fact separability tests for pure, bipartite states, and they have an interesting connection to combinatorics via the cycle index polynomial. Finally, in Section~\ref{sec:comparison}, we analyze the resources needed to implement these tests on quantum computers; in doing so, we show that simpler groups can give comparable performance for fewer resources. Finally, we conclude in Section~\ref{sec:conclusion} with a summary, followed by open questions for future work.

\section{Bipartite Pure-State Separability Test}\label{formula}

Let us begin by reviewing the construction of the bipartite pure-state separability test in \cite{margo}. As discussed therein, it can also be viewed as a $G$-Bose symmetry test. We now recall the definition of a $G$-Bose symmetric state.
\begin{definition}
Let $G$ be a finite group with a unitary representation $U_S:G\to U(\mathcal{H})$, where $U(\mathcal{H})$ denotes the set of all unitaries that act on a Hilbert space $\mathcal{H}$. Then a state $\rho_S$ is called $G$-Bose symmetric if
\begin{align}
    \label{PiG}
    \Pi^G_S\rho_S\Pi^G_S=\rho_S,
\end{align}
where $\Pi^G_S\coloneqq\frac{1}{\lvert G\rvert}\sum_{g\in G}U_S(g)$ is the group representation projection and $\lvert G\rvert$ is the order of the group (the number of elements in the underlying set).
\end{definition}

In Dirac's notation, a pure state $\ket{\psi}_S$ is $G$-Bose symmetric if $\Pi^G_S\ket{\psi}_S=\ket{\psi}_S$. This property is equivalent to $\left\Vert\Pi^G_S\ket{\psi_S}\right\Vert_2=1$, where $\|\cdot\|_2$ denotes the standard Euclidean 2-norm. For a general mixed state, the condition $\tr[\Pi^G_S \rho_S]=1$ is equivalent to the definition in \eqref{PiG}, as argued in \cite{margo}. Thus, the algorithm from \cite{margo} tests for $G$-Bose symmetry by checking the latter condition (see also \cite[Chapter~8]{harrow2005applications}). Indeed, we append an ancillary register $C$ in the state $\ket{0}_C$ to the state $\ket{\psi}_S$ (where the symbol `0' is identified with the identity element of the group~$G$) and act on $C$ with the quantum Fourier transform, leaving a composite system in the state
\begin{align}
    \ket{+}_C\otimes\ket{\psi}_S , 
    \label{eq:init-state}
\end{align}
where $\ket{+}_C$ is defined to be the uniform superposition over all labels of the group elements:
\begin{equation}
\ket{+}_C \coloneqq  \frac{1}{\sqrt{|G|}} \sum_{g\in G}\ket{g}_C.
\label{eq:unif-sup-state}
\end{equation}
Next we act on the state in \eqref{eq:init-state} with the controlled unitary  $\sum_{g\in G}\ket{g}\!\!\bra{g}_C\otimes U_S(g)$, producing the state
\begin{align}
    \frac{1}{\sqrt{|G|}}\sum_{g\in G}\ket{g}_C\otimes(U_S(g)\ket{\psi}_S).
\end{align}
We now project register $C$ onto the state $\ket{+}_C$, which can be accomplished probabilistically by applying the inverse quantum Fourier transform to $C$, measuring $C$ in the $\{\ket{g}\!\!\bra{g}_C\}_{g\in G}$ basis, and declaring `accept' if $\ket{0}\!\!\bra{0}_C$ occurs and  `reject' otherwise. Then the acceptance probability $p$ is given by
\begin{align}
    p&=\bigg\lVert(\bra{+}_C\otimes I_S)\bigg(\frac{1}{\sqrt{|G|}}\sum_{g\in G}\ket{g}_C\otimes(U_S(g)\ket{\psi}_S)\bigg)\bigg\rVert_2^2\\
    &=\bigg\lVert\frac{1}{|G|}\sum_{g\in G}U_S(g)\ket{\psi}_S\bigg\rVert_2^2\\
    &=\left \|\Pi_S^G\ket{\psi}_S\right \|_2^2\\
    &=\tr[\Pi_S^G\ket{\psi}\!\!\bra{\psi}_S].
\end{align}
By convexity, this result is easily generalized to the case in which the state of system~$S$ is described by a density operator $\rho_S$, so that the acceptance probability is equal to $\tr[\Pi_S^G\rho_S]$. 

Now, to test for the separability of a bipartite pure state $\psi_{AB}$, we suppose that $k$ copies of the state $\psi_{AB}$ are available, which we write as $\psi_{AB}^{\otimes k}$. We also identify the $A$ systems by $A_1 \cdots A_k$ and the $B$ systems by $B_1 \cdots B_k$. We then perform an $S_k$-Bose symmetry test on the state $\psi_{AB}^{\otimes k}$ by identifying $S$ with $A_1B_1\cdots A_kB_k$ and $U_S(\pi)$ with $I_{A_1\cdots A_k}\otimes W_{B_1\cdots B_k}(\pi)$, where $\pi\in S_k$ and $W_{B_1\cdots B_k}:S_k\to U(\mathcal{H}_{B_1\cdots B_k})$ is the standard unitary representation of $S_k$ that acts on $\mathcal{H}_{B_1\cdots B_k} \equiv \mathcal{H}_{B_1} \otimes \cdots \otimes \mathcal{H}_{B_k}$ by permuting the Hilbert spaces according to the corresponding permutation. Define $\rho_B\coloneqq \tr_A[\psi_{AB}]$. By applying \eqref{PiG}, the acceptance probability for the bipartite pure-state separability algorithm is given by
\begin{align}
p^{(k)} \coloneqq \tr[\Pi_{B_1\cdots B_k}\rho_B^{\otimes k}]
\label{eq:def-acc-prob}
\end{align}
where
\begin{align}
\Pi_{B_1\cdots B_k}\coloneqq \frac{1}{k!}\sum_{\pi\in S_k}W_{B_1\cdots B_k}(\pi).
\end{align}

\begin{figure}[ptb]
\begin{center}
\includegraphics[
width=4in
]{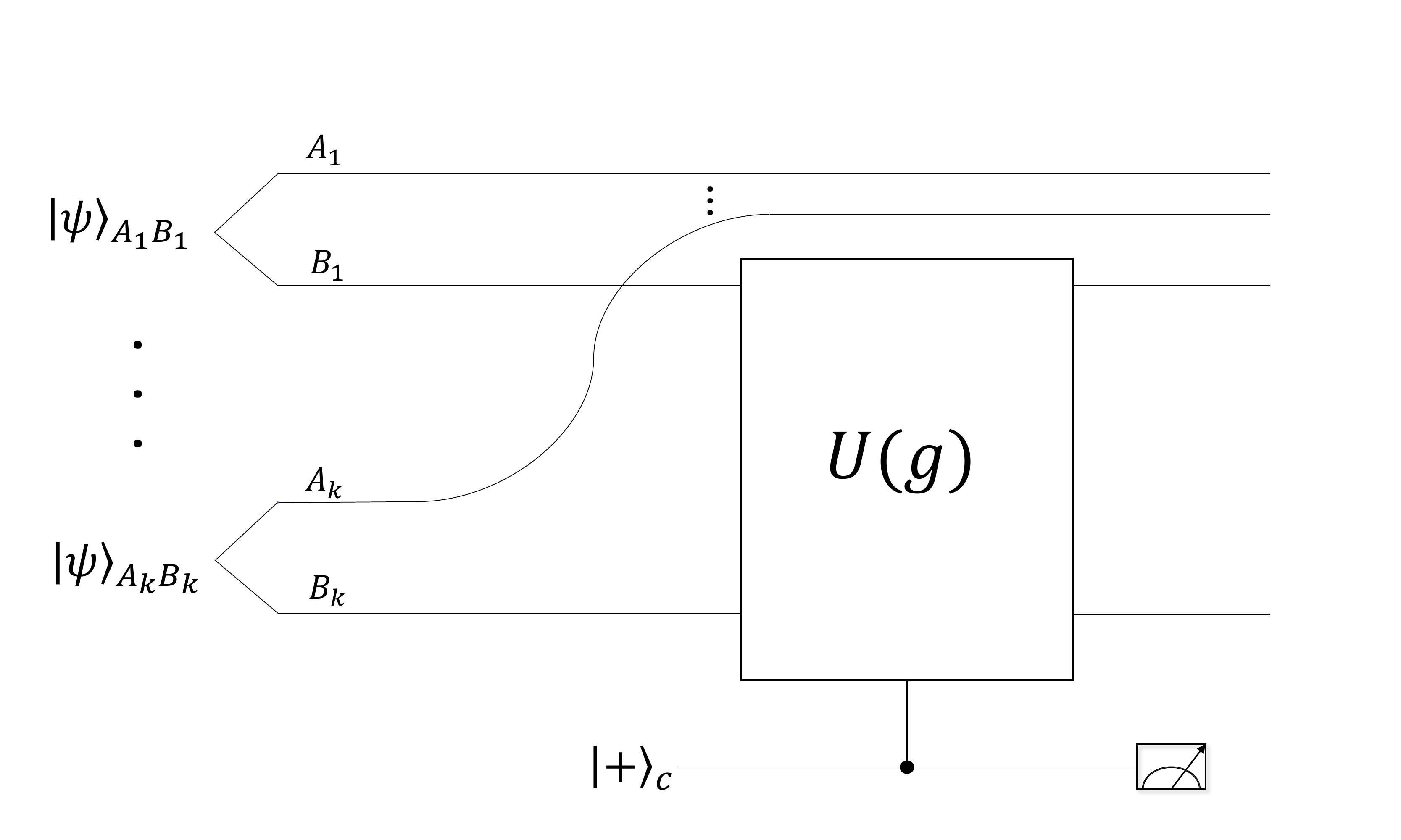}
\end{center}
\caption{Quantum circuit to implement a $G$-Bose symmetry test. We take $k$ copies of an initial bipartite state $\ket{\psi}_{AB}$ and consider the reduced state $\rho_B = \operatorname{Tr}_A [|\psi\rangle\!\langle\psi |_{AB}]$. The collection of these reduced states are then subjected to the separability test determined by the group, where $\ket{+}_C $ is defined in \eqref{eq:unif-sup-state} and $U(g)$ is an element of the group representation.}
\label{fig:circuit}%
\end{figure}

Figure~\ref{fig:circuit} reviews the $G$-Bose symmetry test. The circuit begins with $k$ copies of an initial bipartite state $\ket{\psi}_{AB}$. The $B_i$ subsystems are collected and subject to a controlled unitary gate whose mathematical description involves each unitary $U(g)$. The control register is initialized to the state $\ket{+}_C$, as defined in \eqref{eq:unif-sup-state}. Under the separability test circuit, $G = S_k$, and the unitary representation is a permutation of the $B_i$ subsystems.

Our first main result is a formula for the acceptance probability $p^{(k)}$ in \eqref{eq:def-acc-prob} as a sum over the partitions of $k$ of a product of traces of $\rho_B$ and its powers and certain scaling factors. In particular, the formula is identical to that of the cycle index polynomial of the symmetric group $S_k$, with each variable $x_j$ taking the value $\tr[\rho^j]$ (see the discussion after the following proposition for more on this). In what follows, we use the standard cycle notation wherein $(i_1\ i_2\ \cdots\ i_k)$ refers to the permutation $\sigma:\{1,\ldots,k\}\to\{1,\ldots,k\}$ that sends $i_1$ to $i_2$, $i_2$ to $i_3$, and so on with $i_k$ mapping to $i_1$. For example, $(1\ 2\ 3)$ refers to the permutation $\sigma$ defined by $\sigma(1)=2$, $\sigma(2)=3$, and $\sigma(3)=1$. We follow the standard convention for functions so that a product of cycles (a composition of functions) is read from right to left. For example, $(1\ 2\ 3)(2\ 3)=(1\ 2)(3)=(1\ 2)$.

\begin{theorem}
\label{formulathm}
Let $\psi_{AB}$ denote a pure bipartite state and define $\rho_B\coloneqq \tr_A[\psi_{AB}]$. Then the acceptance probability $p^{(k)}$ for the bipartite pure-state separability test is given by
\begin{align}
p^{(k)}=\sum_{a_1+2a_2+\cdots +ka_k=k}\prod_{j=1}^k\frac{(\tr[\rho_B^j])^{a_j}}{j^{a_j}a_j!},
\label{eq:main-thm-cycle-index-woohoo}
\end{align}
where the sum is taken over the partitions of $k$.
\end{theorem}

\begin{proof}
Let $\pi\coloneqq (1\ 2\ \cdots\ k)$ and consider the representation $W_{B_1\cdots B_k}(\pi)$. It was shown in \cite{Ekert} that $\tr[W_{B_1\cdots B_k}(\pi)\rho_B^{\otimes k}]=\tr[\rho_B^k]$, but we include a proof here for completeness. Expanding $\rho$ in the standard basis as $\rho = \sum_{i,j} p_{i,j} |i\rangle\!\langle j|$, we have
\begin{align}
& \tr[W_{B_1\cdots B_k}(\pi)\rho_B^{\otimes k}] \notag \\
&= \tr\bigg[W_{B_1\cdots B_k}(\pi) \sum_{\substack{i_1,\ldots,i_k \\ j_1,\ldots,j_k}}p_{i_1j_1}\cdots p_{i_kj_k}|i_1\rangle\!\langle j_1|\otimes|i_2\rangle\!\langle j_2|\otimes\cdots\otimes|i_{k}\rangle\!\langle j_k|\bigg]\\
&= \tr\bigg[\sum_{\substack{i_1,\ldots,i_k \\ j_1,\ldots,j_k}}p_{i_1j_1}\cdots p_{i_kj_k}|i_k\rangle\!\langle j_1|\otimes|i_1\rangle\!\langle j_2|\otimes\cdots\otimes|i_{k-1}\rangle\!\langle j_k|\bigg]\\
&=\sum_{\substack{i_1,\ldots,i_k \\ j_1,\ldots,j_k \\ t_1,\ldots,t_k}}p_{i_1j_1}\cdots p_{i_kj_k}\delta_{t_1i_k}\delta_{j_1t_1}\cdots\delta_{t_ki_{k-1}}\delta_{j_kt_k}\\
&=\sum_{t_1,\ldots,t_k}p_{t_2t_1}p_{t_3t_2}\cdots p_{t_kt_{k-1}}p_{t_1t_k}.
\end{align}
Meanwhile, 
\begin{align} \tr[\rho^k]&=\tr\bigg[\sum_{i_1,\ldots,i_k,j_1,\ldots,j_k}p_{i_1j_1}\cdots p_{i_kj_k}|i_1\rangle\!\langle j_1|i_2\rangle\!\langle j_2|\cdots|i_k\rangle\!\langle j_k|\bigg]\\
&=\sum_{i_1,i_2,\ldots,i_k}p_{i_1i_2}p_{i_2i_3}\cdots p_{i_{k-1}i_k}p_{i_ki_1}.
\end{align}
Thus, by relabeling the indices, we see that
\begin{align} \tr[W_{B_1\cdots B_k}(\pi)\rho_B^{\otimes k}]=\tr[\rho^k].
\end{align}
Similarly, we can show for every $m$-cycle $\pi_m\in S_k$ that
\begin{align}
\tr[W_{B_1\cdots B_k}(\pi_m)\rho_B^{\otimes k}]=\tr[\rho^m].
\end{align}
Now suppose $\pi_m$ and $\pi_n$ are disjoint $m$ and $n$-cycles, respectively. Then they act on different Hilbert spaces and so the trace of the product of their representations acting on $\rho_B^{\otimes k}$ splits into the product of traces. That is,
\begin{align} \tr[W_{B_1\cdots B_k}(\pi_m)W_{B_1\cdots B_k}(\pi_n)\rho^{\otimes k}]=\tr[\rho^m]\tr[\rho^n].
\end{align}
Now, since every $m$-cycle yields a factor of $\tr[\rho^m]$ and products of disjoint cycles split the trace, we have
\begin{align} p^{(k)}&=\tr[\Pi_{B_1\cdots B_k}\rho_B^{\otimes k}]\\
&=\tr\bigg[\frac{1}{k!}\sum_{\pi\in S_k} W_{B_1\cdots B_k}(\pi)\rho_B^{\otimes k}\bigg]\\
&=\frac{1}{k!}\sum_{\pi\in S_k}\tr[W_{B_1\cdots B_k}(\pi)\rho_B^{\otimes k}]\\
&=\frac{1}{k!}\sum_{a_1+2a_2+\cdots+ka_k=k}c(a_1,\ldots,a_k)\prod_{j=1}^k\tr[\rho_B^j]^{a_j}
\end{align}
where $c(a_1,\ldots,a_k)$ is the number of cycles in $S_k$ with cycle type $(a_1,\ldots,a_k)$, which is known to be $\frac{k!}{\prod_{j=1}^kj^{a_j}a_j!}$ (see \cite[Eq.~(13.3)]{van2001}). Thus, the equality in \eqref{eq:main-thm-cycle-index-woohoo} follows.
\end{proof}

The cycle index polynomial of a permutation group $G$ is defined by
\begin{align} \label{cycleindex}
    Z(G)(x_1,\ldots,x_n)\coloneqq\frac{1}{|G|}\sum_{g\in G}x_1^{c_1(g)}\cdots x_n^{c_n(g)},
\end{align}
where $c_j(g)$ denotes the number of cycles of length $j$ in the disjoint cycle decomposition of $g$.
Setting $x_j=\tr[\rho_B^j]$, we see that the acceptance probability of the separability test is given by the cycle index polynomial of the symmetric group $S_k$ (see \cite[Chapter 37, pg.~526]{van2001}), so that it satisfies the recurrence relation
\begin{align} \label{recurrence}
    p^{(k)}&=\frac{1}{k}\sum_{j=1}^k\tr[\rho^j]p^{(k-j)}.
\end{align} Furthermore, the cycle index polynomial of the symmetric group $S_k$ is equivalent to
\begin{align}\label{bellpolyAP} \frac{1}{k!}B_k(x_1,x_2,2!\, x_3,3!\, x_4,\ldots,(k-1)!\, x_k),
\end{align}
where $B_k(x_1,\ldots,x_k)$ is the complete Bell polynomial \cite{comtet1974}. From this perspective, the acceptance probability can be interpreted as the $k^{\text{th}}$ raw moment of a probability distribution with the first $k$ cumulants given by $1$, $\tr[\rho^2]$, \ldots, $\tr[\rho^k]$.

Using a variant of Newton's identities \cite{Newton}, we can also write the acceptance probability as both a determinant and a permanent of sequences of matrices. To see this, define two sequences of matrices by
\begin{align}
P_k\coloneqq\begin{pmatrix}
1&1&0&\cdots\\
\tr[\rho^2]&1&2&0&\cdots\\
\vdots&\vdots&\ddots&\ddots\\
\tr[\rho^{k-1}]&\tr[\rho^{k-2}]&\cdots&1&k-1\\
\tr[\rho^k]&\tr[\rho^{k-1}]&\cdots&\tr[\rho^2]&1
\end{pmatrix}
\end{align}
and
\begin{align}
D_k\coloneqq\begin{pmatrix}
1&-1&0&\cdots\\
\tr[\rho^2]&1&-2&0&\cdots\\
\vdots&\vdots&\ddots&\ddots\\
\tr[\rho^{k-1}]&\tr[\rho^{k-2}]&\cdots&1&1-k\\
\tr[\rho^k]&\tr[\rho^{k-1}]&\cdots&\tr[\rho^2]&1
\end{pmatrix}.
\end{align}
Let $\{\lambda_i\}_i$ be the eigenvalues of $\rho_B$. Then $\{\lambda_i^m\}_i$ are the eigenvalues of $\rho_B^m$. So the trace of $\rho_B^m$ is given by
\begin{align} \tr(\rho_B^m)=\sum_i\lambda_i^m.
\end{align}
That is, $\tr(\rho_B^m)$ is the $m$-th power sum of the eigenvalues, so that the acceptance probability $p^{(k)}$ is given by the $k$-th complete homogeneous symmetric polynomial. A variant of Newton's identities \cite{Newton} relates the $k$-th complete homogeneous symmetric polynomial $h_k$ with the $k$-th power sum $p_k$ by $kh_k=\sum_{j=1}^kh_{k-j}p_j$. Now applying Cramer's rule, we can write
\begin{align} \label{det}
h_k&=\frac{1}{k!}\det\begin{pmatrix}
p_1&-1&0&\cdots\\
p_2&p_1&-2&0&\cdots\\
\vdots&\vdots&\ddots&\ddots\\
p_{k-1}&p_{k-2}&\cdots&p_1&1-k\\
p_k&p_{k-1}&\cdots&p_2&p_1
\end{pmatrix}\\
&\label{perm}=\frac{1}{k!}\perm\begin{pmatrix}
p_1&1&0&\cdots\\
p_2&p_1&2&0&\cdots\\
\vdots&\vdots&\ddots&\ddots\\
p_{k-1}&p_{k-2}&\cdots&p_1&k-1\\
p_k&p_{k-1}&\cdots&p_2&p_1
\end{pmatrix}.
\end{align}
Thus, by identifying the acceptance probability with $h_k$ and the trace of $\rho^j$ with $p_j$, we have
\begin{align} p^{(k)}=\frac{1}{k!}\perm(P_k)=\frac{1}{k!}\det(D_k).
\end{align}
See \cite[Chapter~1, Section~2]{macdonald1995} for more information on Eqs.~\eqref{recurrence}, \eqref{det}, and \eqref{perm}.

\section{Strictly Decreasing Acceptance Probability}

\label{sec:conjecture}

In \cite{margo}, it was conjectured that the acceptance probability of the bipartite pure-state separability test is monotone non-increasing in $k$. We answer this conjecture in the affirmative as a corollary of the following lemma about complete Bell polynomials. In fact, we show that this inequality is strict and the acceptance probability approaches zero in the limit $k\to\infty$ whenever $\rho_B$ is not a pure state.

\begin{lemma}\label{bellpoly} Suppose $1=p_1\ge\cdots\ge p_{n+1}\ge0$. Then
\begin{align} \frac{1}{(n+1)!}B_{n+1}(p_1,\ldots,n!p_{n+1})\le\frac{1}{n!}B_n(p_1,\ldots,(n-1)!p_n).
\end{align}
\end{lemma}
\begin{proof} The complete Bell polynomial satisfies the recurrence relation \cite{bell1934exponential}
\begin{align} B_{n+1}(x_1,\ldots,x_{n+1})&=\sum_{j=0}^n\binom{n}{j}B_{n-j}(x_1,\ldots,x_{n-j})x_{j+1}.
\end{align}
Letting $x_i=(i-1)!\ p_i$ and defining $B_j\coloneqq B_j(1,\ldots,(j-1)!p_j)$, we have
\begin{align}
&\frac{1}{n!}B_n-\frac{1}{(n+1)!}B_{n+1}\notag \\
&=\frac{1}{n!}B_n-\frac{1}{(n+1)!}\sum_{i=0}^n \binom{n}{i}B_{n-i}i!p_{i+1}\\
&=\frac{n}{(n+1)!}B_n-\frac{1}{(n+1)!}\sum_{i=1}^n \binom{n}{i}B_{n-i}i!p_{i+1}\\
&=\frac{n}{(n+1)!}B_n-\frac{n}{(n+1)!}\sum_{i=0}^{n-1} \binom{n-1}{i}B_{n-i-1}i!p_{i+2}\\
&\label{inequalityjustification}\ge\frac{n}{(n+1)!}B_n-\frac{n}{(n+1)!}\sum_{i=0}^{n-1} \binom{n-1}{i}B_{n-i-1}i!p_{i+1}\\
&=\frac{n}{(n+1)!}B_n-\frac{n}{(n+1)!}B_n\\
&=0,
\end{align}
where in the second equality, the $i=0$ term was combined with the first term, and the inequality is justified by $p_{i+2}\le p_{i+1}$. 
\end{proof}

\begin{proposition}\label{th:decreasing} The acceptance probability $p^{(k)}$ is strictly decreasing and  $\lim_{k\to\infty}p^{(k)}=0$ when $\rho_B$ is not a pure state.
\end{proposition}

\begin{proof}
Writing the acceptance probability in terms of the complete Bell polynomials as in \eqref{bellpolyAP} and noting that $1=\tr[\rho_B]>\tr[\rho_B^2]>\cdots>\tr[\rho_B^{k+1}]>0$ (since the eigenvalues of a mixed state are strictly less than one), we have
\begin{align}
    p^{(k+1)}&=\frac{1}{(k+1)!}B_{k+1}(1,\ldots,k!\tr[\rho_B^{k+1}])\\
    &<\frac{1}{k!}B(1,\ldots,(k-1)!\tr[\rho_B^k])\\
    &=p^{(k)},
\end{align}
where we have applied Lemma \ref{bellpoly}, noting that the inequality \eqref{inequalityjustification} is now strict because the inequalities between traces of powers of $\rho_B$ are strict as noted above. Thus, $p^{(k)}$ is strictly decreasing. Since the acceptance probability is given by the cycle index polynomial of the symmetric group, it satisfies the recurrence
\begin{align}
    p^{(k)}&=\frac{1}{k}\sum_{j=1}^k\tr[\rho^j]p^{(k-j)}.
\end{align}
Let $r$ be the total number of eigenvalues of $\rho_B$. From the fact that $p^{(k-j)}\le1$, we have
\begin{align} p^{(k)}&\le\frac1k\sum_{j=1}^k\tr(\rho_B^j)\\
&=\frac1k\sum_{j=1}^k\sum_{i=1}^r\lambda_i^j\\
&=\sum_{i=1}^r\frac1k\sum_{j=1}^k\lambda_i^j\\
&=\sum_{i=1}^r\frac1k\bigg(\frac{1-\lambda_i^{k+1}}{1-\lambda_i}-1\bigg).
\end{align}
Since $\rho_B$ is not a pure state, we have $\lambda_i<1$ for all $i$, so that $\lambda_i^{k+1}\to 0$ as $k\to \infty$ and the acceptance probability therefore goes to zero in the $k\to\infty$ limit.
\end{proof}

These results indicate that as $k\to\infty$, fewer repetitions of the test are needed to determine whether a given pure state is entangled. There is a trade-off, however, between increasing $k$ and the computational resources needed to conduct a single test. As $k$ increases, one might suspect that the resources needed will increase in such a way that a large enough $k$ is not feasible. Indeed, as one of our results, we discuss the scaling in this claim in Section~\ref{sec:comparison}.

\section{Generalization of the Algorithm}

\label{generalization}

It is possible to consider a generalization of the bipartite pure-state separability algorithm to groups other than the symmetric group. Furthermore, we can show that these algorithms are also separability tests. Let $G$ be a finite group, and let $\psi_{AB}$ be a pure state. Recall that Cayley's theorem \cite{Dummit_Foote} guarantees that every finite group is isomorphic to a subgroup of a permutation group. Then there is a representation of $G$ that takes every $g\in G$ to an element $\pi\in S_k$ for some $k\in\mathbb{N}$ and maps $\pi$ to the operator that permutes the Hilbert spaces in the composite Hilbert space $\mathcal{H}^{\otimes k}$. Then a generalization of the bipartite pure-state separability algorithm is given by performing a $G$-Bose symmetry test on the state $\psi_{AB}^{\otimes k}$.

By the argument in the proof of Theorem~\ref{formulathm}, we see that one simply has to count the number of cycles of any given cycle type in the permutation subgroup isomorphic to $G$ to obtain a formula for the acceptance probability of the algorithm. That is, the argument in Theorem~\ref{formulathm}, combined with Cayley's theorem, proves the following theorem:

\begin{theorem}
\label{generalformula}
Let $p_G$ denote the acceptance probability with respect to the group $G$ for the generalization of the bipartite pure-state separability algorithm. Then
\begin{align} \label{pG} p_G=Z(G)(1,\ldots,\tr[\rho^k]).
\end{align}
\sloppy That is, the acceptance probability $p_G$ is given by the cycle index polynomial \eqref{cycleindex} of~$G$ evaluated at $x_j=\tr[\rho^j]$ for $j \in \{1,\ldots,k\}$.
\end{theorem}

As an aside, we note that \eqref{pG} has an interesting combinatorial meaning. Let $\{\lambda_i\}_{i=1}^r$ denote the eigenvalues of $\rho$. By Pólya's enumeration theorem \cite{brualdi2010,tucker1995}, we can interpret \eqref{pG} as a generating function for the number of nonequivalent colorings of a set $S$ with the $r$ colors $\{\lambda_i\}_{i=1}^r$. The role of $G$ here is to define the equivalence between colorings through its action on $S$.

We now list several examples involving finite groups, including the symmetric groups, which reproduces the result obtained in Section~\ref{formula} (the acceptance probability of the bipartite pure-state separability test). The most trivial example is the identity group.

\begin{example}
Let $G=\{e\}$ be the identity group and let $W_{B_1\cdots B_k}:G\to U(\mathcal{H}_{B_1\cdots B_k})$ be the unitary representation mentioned above. Then $W_{B_1\cdots B_k}(e)=I$, and the acceptance probability is given by the cycle index polynomial of the identity group. That is,
\begin{align}
p_{\text{id}}^{(k)}=\tr[\rho_B]^k=1.
\end{align}
\qed
\end{example}

\begin{example}
Let $G=C_k=\langle g\rangle$ be the cyclic group of rotations of a regular $k$-gon. We have many options for the representation, as there are in general many $k$-cycles to which we can map the generator of $C_k$. However, each choice produces an isomorphic permutation subgroup, and so we may as well choose $W_{B_1\cdots B_k}(g)=(1\ 2\ \cdots\ k)$. Then the acceptance probability is given by the cycle index polynomial of $C_k$. That is,
\begin{align} p_{\text{cyc}}^{(k)}=\frac1k\sum_{m|k}\phi(m)(\tr[\rho_B^m])^{k/m},
\end{align}
where the sum is over all $m$ dividing $k$, and $\phi(m)$ is the Euler $\phi$-function \cite{Dummit_Foote}, which gives the number of $n\in\mathbb{N}$ such that $n\le m$ and gcd$(n,m)=1$.
\qed
\end{example}

\begin{example}
Let $G=D_k$ be the dihedral group, which is generated by the rotation~$r$ and reflection $f$ of a regular $k$-gon. As for the case of the cyclic group, we choose a unitary representation so that it sends the rotation $r$ to a $k$-cycle. The reflection $f$ is then mapped to a product of $\lfloor{k/2}\rfloor$-cycles. Then the acceptance probability is given by the cycle index polynomial of $D_k$. That is, when $k$ is even,
\begin{align}
p_{\text{dihedral}}^{(k)}&=\frac{1}{2k}\sum_{m|k}\phi(m)(\tr[\rho_B^m])^{k/m} +\frac14\bigg((\tr[\rho_B^2])^{(k-2)/2}+(\tr[\rho_B^2])^{k/2}\bigg)\\
&=\frac12 p^{(k)}_{\text{cyc}}+\frac14\bigg((\tr[\rho_B^2])^{(k-2)/2}+(\tr[\rho_B^2])^{k/2}\bigg),
\end{align}
and when $k$ is odd,
\begin{align}p_{\text{dihedral}}^{(k)}&=\frac{1}{2k}\sum_{m|k}\phi(m)(\tr[\rho_B^m])^{k/m}+\frac12(\tr[\rho_B^2])^{(k-1)/2}\\
&=\frac12 p^{(k)}_{\text{cyc}}+\frac12(\tr[\rho_B^2])^{(k-1)/2}.
\end{align}
\qed
\end{example}

\begin{example}
Let $G=A_k$ be the alternating group, which is already a permutation group. Then the acceptance probability is given by the cycle index polynomial of $A_k$. That is,
\begin{align}
p_{\text{alt}}^{(k)}=\sum_{a_1+2a_2+\cdots +ka_k=k}\prod_{j=1}^k\frac{1+(-1)^{a_2+a_4+\cdots}}{j^{a_j}a_j!}(\tr[\rho_B^j])^{a_j},
\end{align}
where $(-1)^{a_2+a_4+\cdots}$ denotes $(-1)^{a_2+a_4+\cdots+a_k}$ if $k$ is even and $(-1)^{a_2+a_4+\cdots+a_{k-1}}$ if $k$ is odd.
\qed
\end{example}

\begin{example}
We also list the example already discussed in Theorem~\ref{formulathm}. Let $G=S_k$ be the symmetric group, which is already a permutation group. Then the acceptance probability is given by the cycle index polynomial of $S_k$. That is,
\begin{align}p_{\text{sym}}^{(k)}=\sum_{a_1+2a_2+\cdots +ka_k=k}\prod_{j=1}^k\frac{(\tr[\rho_B^j])^{a_j}}{j^{a_j}a_j!}.
\end{align}
\qed
\end{example}

\begin{example} \sloppy Let $G=Q_8$ be the quaternion group. We can represent $Q_8$ as a subgroup of $S_8$ by identifying $\{1,-1,i,-i,j,-j,k,-k\}$ with 
$\{e,(1\ 2)(3\ 4)(5\ 6)(7\ 8),(1\ 3\ 2\ 4)(5\ 7\ 6\ 8),(1\ 4\ 2\ 3)(5\ 8\ 6\ 7),(1\ 5\ 2\ 6)(3\ 8\ 4\ 7),$
$(1\ 6\ 2\ 5)(3\ 7\ 4\ 8),(1\ 7\ 2\ 8)(3\ 5\ 4\ 6),(1\ 8\ 2\ 7)(3\ 6\ 4\ 5)\}$. Then the acceptance probability is given by the cycle index polynomial of $Q_8$, which can be read off from the permutation representation. That is,
\begin{align} p_{\text{quat}}=\frac18\bigg(1+(\tr[\rho_B^2])^4+6(\tr[\rho_B^4])^2\bigg).
\end{align}
\qed
\end{example}

\begin{example} 
In this example, we generalize the cyclic test to products of cyclic groups. Let $G=\mathbb{Z}_m^k$ be the product of $k$ copies of the group $\mathbb{Z}_m$. We represent $G$ as a permutation subgroup by labeling its elements and letting them act on the group to construct a permutation. For example, if $k=1$, then $G=\{0,1,\ldots,m-1\}$. Since 0 has no effect on any element of the group, we map it to the identity element $e$. Meanwhile, 1 acts on each element of the group by sending 0 to 1, 1 to 2, and so on. So we identify 1 with the cycle $(1\ \cdots\ m)$. The remaining permutations are defined similarly. 

Returning to the more general setting, we see that the elements of each order $n$ correspond to products of $n$-cycles. Now, for an element of $G$ to have order $n$, each component must contain an element of an order that divides $n$, with at least one component filled by an element of order $n$. So the number of elements of order $n$ is given by
\begin{align} \sum_{i=1}^k\binom{k}{i}(\phi(n))^i\bigg(\sum_{\substack{l|n\\ l<n}}\phi(l)\bigg)^{k-i}
&=\bigg(\phi(n)+\sum_{\substack{l|n\\ l<n}}\phi(l)\bigg)^k-\bigg(\sum_{\substack{l|n\\ l<n}}\phi(l)\bigg)^k\\
&=n^k-(n-\phi(n))^k,
\end{align} 
where $\phi$ denotes the Euler $\phi$-function. The acceptance probability given by the cycle index polynomial of $\mathbb{Z}_m^k$ is then
\begin{align} p_{\mathbb{Z}_m^k}^{(k)}=\frac{1}{m^k}\sum_{n|m}(n^k-(n-\phi(n))^k)(\tr[\rho_B^n])^{\frac{m^k}{n}}.
\end{align}
\qed
\end{example}

Finally, we show that the above nontrivial examples, as well as any other example involving a nontrivial finite group, are tests for separability of a pure bipartite state. Thus, we have produced an entire class of separability tests.

\begin{proposition}
Let $\psi_{AB}$ denote a pure bipartite state. Then the generalized bipartite pure-state separability algorithm is, in fact, a faithful test for separability of $\psi_{AB}$ for any nontrivial finite group $G$, meaning that the acceptance probability is equal to one if and only if the pure state is a separable state.
\end{proposition}

\begin{proof}
Suppose $\psi_{AB}$ is separable. That is, $\ket{\psi}_{AB}=\ket{\phi}_A\otimes\ket{\varphi}_B$ for some states $\ket{\phi}_A\in\mathcal{H}_A$ and $\ket{\varphi}_B\in\mathcal{H}_B$. Then
\begin{equation}
\rho_B \coloneqq \tr_A[\psi_{AB}]
=\tr_A[\ket{\phi}\!\!\bra{\phi}_A\otimes\ket{\varphi}\!\!\bra{\varphi}_B]
=\ket{\varphi}\!\!\bra{\varphi}_B.
\end{equation}
That is, $\rho_B$ is a pure state. From Theorem~\ref{generalformula}, the acceptance probability of the algorithm is given by the cycle index polynomial evaluated at the traces of increasing powers of $\rho_B$. But since $\rho_B$ is pure, $\tr[\rho_B^j]=1$ for all $j\in\{1,\ldots,n\}$. Then the acceptance probability is equal to the cycle index polynomial at $x_j=1$ for all $j\in\{1,\ldots,n\}$. That is,
\begin{align}
p_G&=Z(G)(1,\ldots,\tr[\rho^n])\\
&=Z(G)(1,\ldots,1)\\
&=\frac{1}{|G|}\sum_{g\in G} 1^{c_1(g)}\cdots1^{c_n(g)}\\
&=\frac{1}{|G|}\sum_{g\in G} 1\\
&=1
\end{align}
where $c_i(g)$ denotes the number of cycles of length $i$ in the disjoint cycle decomposition of $g$. Thus, $\psi_{AB}$ separable implies that the acceptance probability is identically one.

Now suppose $\rho_B$ is a mixed state. Then $\tr[\rho_B^j]<1$ for all $j>1$ and we have
\begin{align} p_G&=Z(G)(1,\tr[\rho_B^2],\ldots,\tr[\rho_B^n])\\
&=\frac{1}{|G|}\sum_{g\in G} 1^{c_1(g)}(\tr[\rho_B^2])^{c_2(g)}\cdots(\tr[\rho_B^n])^{c_n(g)}\\
&<\frac{1}{|G|}\sum_{g\in G} 1^{c_1(g)}\cdots1^{c_n(g)}\\
&=\frac{1}{|G|}\sum_{g\in G} 1\\
&=1,
\end{align}
where we have used the assumption that $G$ is nontrivial to guarantee that at least one of the $c_j(g)$ is nonzero so that the inequality holds. Thus, the test is faithful.
\end{proof}

\section{Resource Comparison of Symmetry Tests}

\label{sec:comparison}

Given the generalization in Section~\ref{generalization}, we can now compare the performance of these separability tests. There are two practical concerns to consider when implementing such a test: the rate at which the acceptance probability decays and the resources required to construct it. The cycle index polynomial results described above allow for direct analysis of the former topic, but the latter requires additional consideration before it can be adequately addressed. First, we will specify how resources are counted for each algorithm. Then we compare the resource cost for each algorithm given this framework. We accompany this with a discussion of the acceptance probability of the compared methods.

We now clarify what is meant by resources in this context. For the $G$-Bose symmetry test described in \cite{margo} and tests of that nature, the two primary resources are the number of gates used to construct the test and how many qubits are needed in the control register. We begin with a discussion of gate counting.

\subsection{Resource Counting of Quantum Gates}

The unitary representation in this context is always formed from a collection of SWAP gates used to permute the subsystems. SWAP gates can be realized by a sequence of three CNOT gates in alternating direction. Often, the literature commonly counts the number of CNOT gates used as a resource (see, e.g., \cite{grassl2000cyclic})
; however, particular architectures may have more efficient realizations of the SWAP gate. Furthermore, this algorithm actually calls for controlled-SWAP gates, which may have vastly different compilations between architectures. For the purposes of this discussion, we will be counting the necessary number of controlled-SWAPs alone. Furthermore, we do not restrict to swapping between consecutive Hilbert spaces, although in principle this could be a limitation of particular systems.

Here, we give an explicit construction for two example groups. The first is the cyclic group test, which is a simple Abelian subgroup of the symmetric group and therefore of interest as a point of comparison. Although constructions of cyclic shifts exist in the literature, our construction follows binary encoding procedures \cite{babbush2018encoding,low2019hamiltonian} and uses fewer gates than a na\"ive implementation and thus warrants discussion. The second construction given describes a recursive implementation of the full permutation test.  Similarly, although the quantum Schur transform \cite{BCH06,BCH07,krovi2019efficient} gives a recipe for implementing the symmetric group in principle, the gate construction is abstract and thus difficult to use for accurate gate counts compared to other approaches. As such, we utilize the construction given in \cite{barenco1997stabilization}. In the following two subsections, we show that a cyclic group test can be implemented with $\mathcal{O}(k \log (k))$ controlled-SWAP gates and a full symmetric group test (also known as a permutation test) with $\mathcal{O}(k^2)$ controlled-SWAP gates.

Note that our tests are only valid for pure states. If it is unknown whether the input state is pure or mixed, one can first check whether it is indeed pure. To accomplish this, perform a swap test on the state in question to verify that it has purity approximately equal to one and is sufficiently pure.

\subsubsection{Cyclic Group}

Analysis of the cyclic group benefits from established literature. Any cyclic permutation can be achieved in constant depth with $k-1$ gates, where $k$ is the order of the cycle \cite{grassl2000cyclic}. We will now show that any cyclic group test can be generated by implementing solely the elements in that cycle that are powers of two. This means that the resource cost of implementing the cyclic test of order $k$ is $(k-1)\log_2 (k)$, and the constant depth condition above from \cite{grassl2000cyclic} gives a corresponding depth of $\mathcal{O}(\log_2(k))$ in the separability test.

First, recall that the $k$-order cyclic group is isomorphic to the set $\mathbb{Z}_k$ of integers modulo $k$ under addition. This will allow us to symbolically represent each element by a single number, understood in this context to be modulo $k$. 

Since the case of $k=1$ is trivial, let us first consider the base case of $k=2$. This example illustrates the general construction of cyclic tests and recreates the well-established swap test \cite{barenco1997stabilization,buhrman2001quantum}. The controlled-SWAP element corresponds to the element $1 = 2^0$, and is the sole gate needed, and the identity element is naturally $0$. (Note that $1$ is the sole power of two in $\mathbb{Z}_2 = \{0,1\}$.) The control state for this test is given by a single qubit state of
\begin{equation}
    \ket{+}_{C(2)} =\frac{1}{\sqrt{2}}(\ket{0} + \ket{1})\, ,
\end{equation}
where we employ the computational basis. It is clear that each element in the ancillary basis will give rise to its corresponding group element with this test.

How does this construction generalize? For each given $k$, we follow a similar recipe as above. As $C_k$ is isomorphic to $\mathbb{Z}_k$, start by identifying each cycle in $C_k$ with a number in $\mathbb{Z}_k$. If we always map the first $k$-cycle to one, then this map follows simply by mapping cycle composition to integer addition by one. Consider, for instance, the case of $C_5$. Then the first cycle is $(1\ 2\ 3\ 4\ 5)$. Map this to $1$. Then the next element, $(1\ 3\ 5\ 2\ 4)=(1\ 2\ 3\ 4\ 5)(1\ 2\ 3\ 4\ 5)$ maps to $1+1=2$. After we have identified each element of $C_k$ with an element of $\mathbb{Z}_k$, we can always rewrite these numbers in binary. The beauty of binary construction, as is well appreciated in computer science, is that only elements corresponding to powers of two need to be individually defined, and every other number can be generated from combinations of them. Thus, after this second rewrite, we have a mapping between every cycle in $C_k$ and a binary number. Now to construct the circuit, we only need to implement controlled gates that correspond to cycles that have mapped to a power of two. For $C_5$, this would be gates that have mapped to $001$, $010$, and $100$ (in decimal: 1, 2, and 4 respectively). 

To show how this construction grows, it is most convenient to denote the gates by which power of two they implement.  In Figure~\ref{fig:cyclic}, we label gates as $2^j$ where $j$ ranges from 0 to $\lfloor \log_2(k-1)\rfloor$. To see why $\lfloor \log_2(k-1)\rfloor$ is the final gate, recall the convention that $\mathbb{Z}_k$ always contains $0$ instead of $k$. Then the bound falls out from inspection. Revisiting our above example of $k=5$, the gates we identified as necessary can be equivalently represented as $001=2^0$, $010=2^1$, and $100=2^2$.

This construction can also be achieved by considering the labeling of the control state. If the computational basis is read as a number in binary, we can clearly define the relationship between the computational basis and the group element construction as $\ket{g}=\ket{g_{\textrm{binary}}}=\ket{g_{\textrm{decimal}}}$, where the abstract construction is equivalent to a computational basis in binary, which equivalently realizes the familiar group element in decimal. For example, following the above convention, the basis state for $k=5$ given by $\ket{(1\ 3\ 5\ 2\ 4)}=\ket{10}=\ket{2}$ indicates that the element $(1\ 3\ 5\ 2\ 4)$ can be labeled as the $2$ element of the group. As $2$ is obviously a power of $2$, this group element must be encoded in the circuit. This construction is shown generally in Figure~\ref{fig:cyclic} and for our specific example of $k=5$ in Figure~\ref{fig:cyclicexample}. Note that all elements of $C_k$ will take at most $k-1$ SWAP gates to implement. 

Furthermore, note that cyclic permutations can be implemented in a constant depth of two \cite{grassl2000cyclic}. To maintain this depth even for the controlled gates, a GHZ state, $\frac{1}{\sqrt{2}}(\ket{0}^{\otimes m} +\ket{1}^{\otimes m})$, where $m=\lfloor k/2\rfloor$, can be used instead of a single plus state, similar to the approach employed in \cite{yihui}. Then the controls can act on different qubits of the state, and the final measurement is taken by projecting back to the GHZ state. This state preparation and the corresponding measurement may add complexity to the ancilla register; however, since the circuit to prepare a $k$-qubit GHZ state has depth $\mathcal{O}(\log_2(k))$ (with the circuit to project onto it being its inverse), this gives the cyclic group test a depth that grows as $\mathcal{O}(\log_2 (k))$.

To see that this circuit is capable of generating every element of the cyclic group, we again refer to the isomorphism between $C_k$ and $\mathbb{Z}_k$. Writing every element of $\mathbb{Z}_k$ in binary, it becomes obvious that every element can be written as an addition of powers of 2 that form the basis of binary numbers. As such, only elements corresponding to new ``digits" need to be considered.

\begin{figure}[ptb]
\begin{center}
\includegraphics[
width=4.5 in
]{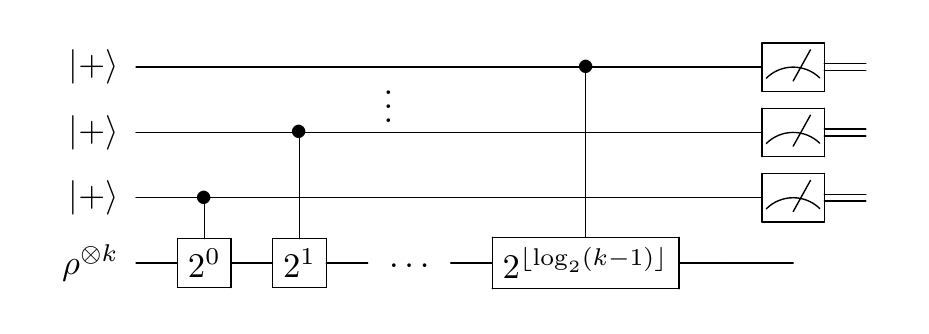}
\end{center}
\caption{Figure demonstrating how to systematically generate a test for the cyclic group of order $k$. The notation $(2^j)$ indicates the unitary representation of the element in $C_k$ labeled by the $j$-th power of two. Alternatively, this element is obtained by the full $k$-cycle $(1,2,\ldots,k)$ acting on itself $2^j$ times. Note that the final power is always given by $\lfloor \log_2(k-1)\rfloor.$ Also, $|+\rangle = (\ket{0}+\ket{1})/\sqrt{2}$ in the circuit diagram above and the final measurements are performed in the Hadamard basis, accepting if all $+1$ outcomes occur.}
\label{fig:cyclic}%
\end{figure}

\begin{figure}[ptb]
\begin{center}
\includegraphics[
width=5 in
]{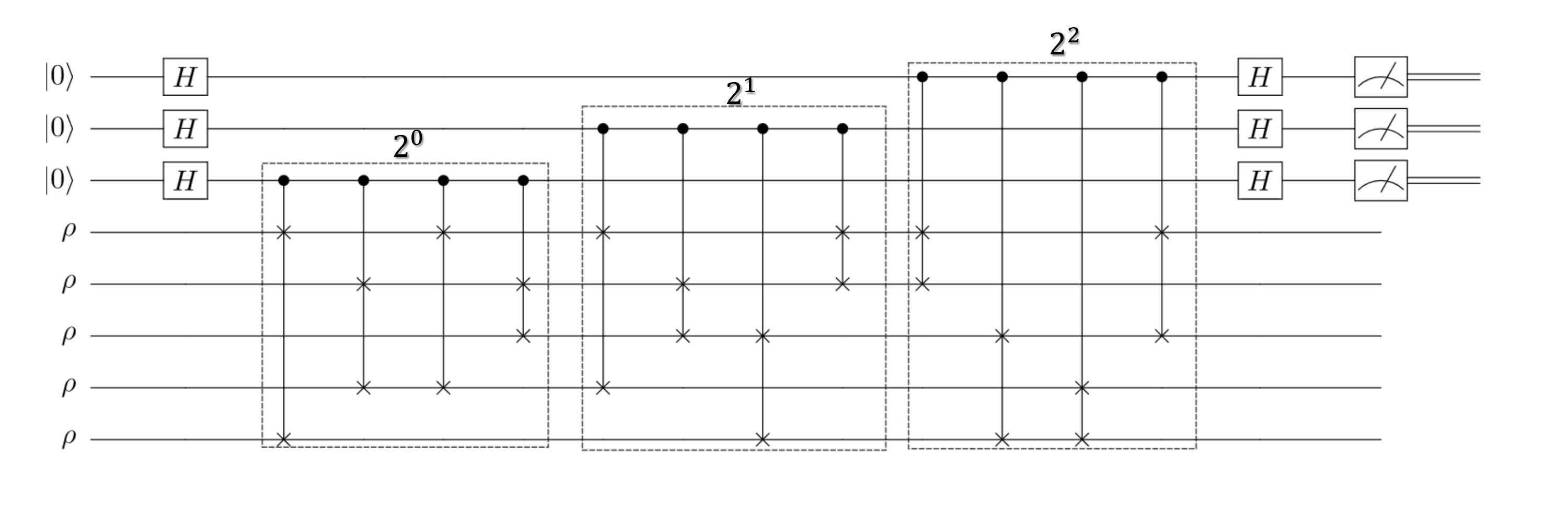}
\end{center}
\caption{An example of the cyclic group test for $k=5$. The notation $(2^j)$ indicates the unitary representation of the element in $C_k$ labeled by the $j$-th power of two. For this case, only the elements corresponding to $2^0$, $2^1$, and $2^2$ contribute. Notice that, if the gates are not controlled on the same qubit, each individual cycle collapses to a depth of two with $k-1$ gates.}
\label{fig:cyclicexample}%
\end{figure}

\subsubsection{Symmetric Group}

We now review a recursive algorithm for the construction of the symmetric group.  Necessary to this construction is the proof that the entire group $S_k$ can be generated in a convenient way, using solely transpositions. This construction is equivalent to that given in \cite{barenco1997stabilization}, but we explain it here for completeness.

Observe that $S_2$ can be generated by the element $(1\ 2)$. To generate $S_3$, we need only act on this element from the left by $(1\ 3)$ and $(2\ 3)$. Indeed, the remaining elements of $S_3$ are given by $(1\ 2\ 3)=(1\ 3)(1\ 2)$ and $(1\ 3\ 2)=(2\ 3)(1\ 2)$. This serves as our base case, and we now proceed by induction. Suppose we can generate every element of $S_{k-1}$ in this way. We must show that the remaining elements of $S_k$ are given by acting on $S_{k-1}$ from the left by the transpositions of the form $(i\ k)$ for $i\in \{1,2,\ldots,k-1\}$. To see this, let $(i_1\ i_2\ \cdots\ i_m)$ be an arbitrary $m$-cycle in $S_{k-1}$. Then acting from the left by $(i_j\ k)$ for some $j\in \{1,\ldots,m\}$ yields $(i_j\ k)(i_1\ i_2\ \cdots\ i_m)=(i_1\ i_2\ \cdots\ i_{j-1}\ k\ i_j\ \cdots\ i_m)$. In this way, we can generate every cycle in $S_k$. Since every element of $S_k$ can be decomposed into a product of disjoint cycles, we can now generate every element of $S_k$ recursively by appending only transpositions of the form $(i\ k)$. We can visualize this construction by the circuit given in Figure~\ref{fig:permutation4} for an example when $k =4$. 

Given a way to generate $S_k$, we now need an appropriate control state to implement these elements. By supposition, the identity can always be implemented via the state $\ket{0}$ tensored with itself to some power. What then for the remaining states? Consider only one 'layer' of the recursive construction of $S_k$. It suffices to only ever use one transposition at a time. Thus the control state for every $i$-th layer of transpositions should take the form
\begin{equation}
    \ket{+}_{S_i}=\frac{1}{\sqrt{i+1}}(\ket{0}^{\otimes i} + \ket{10\cdots 0} + \ket{01\cdots 0}+\cdots +\ket{00\cdots 1})\, ,
\end{equation}
as given in \cite{barenco1997stabilization}. These individual control states should be concatenated together to form the control register for the entire algorithm. For a quick sanity check, when considering the tensor product of such states as $i$ ranges from 1 to $k$, the normalization constant out in front becomes $\sqrt{k!}=\sqrt{|S_k|}$.

However, a question remains; can the control register for such a circuit also be generated recursively? Observe, in Figure~\ref{fig:permutation4}, that we denote a series of gates $A_j$ that act on the control register to create superpositions. Furthermore, notice that we have arranged the transpositions in a consistent manner such that each gate is appended in ascending order of transposition. Then we define the gate $A_j$ to act as such:
\begin{equation}
    A_j \ket{0}^{\otimes j-1} =\frac{1}{\sqrt{j}} (\ket{0}^{\otimes j-1} + \sqrt{j-1}\ket{W_{j-1}}) \, ,
\end{equation}
where $\ket{W_{j-1}}=\frac{1}{\sqrt{j-1}}\sum_{i=1}^{j-1}\ket{2^i}$ is the $W$-state on $j-1$ qubits. Here $\ket{2^i}$ is the state with a one in the $i$-th component and a zero elsewhere. We can observe by inspection that this action, when taken recursively from $j=2$ to $j=k$, will generate a superposition over $k!$ basis elements. An example of this construction can be seen in Figure~\ref{fig:permutation4} for $k=4$. 

There are several choices available to construct these $A_j$ gates. We review two here. One recursive approach is to begin by designing the circuit for $A_i$; then the next gate $A_{i+1}$ is given by adding $i+1$ control qubits, initializing the first qubit to a superposition of $(\frac{1}{\sqrt{i}} \ket{0}+\frac{\sqrt{i-1}}{\sqrt{i}}\ket{1})$, then controlling off of this state, implement $A_i$ on the remaining new qubits. However, this na\"ive approach will use numerous gates and quickly grow in size. In \cite{barenco1997stabilization}, they assume the first $i$ qubits are initialized and then add $i+1$ qubits for the recursive step. The $(i+1)$-th qubit can be acted on by a one-qubit gate $U_i$ given by
\begin{align}
    U_i\coloneqq\frac{1}{\sqrt{i+1}}
    \begin{pmatrix} 
    1 & -\sqrt{i}\\
    \sqrt{i} & 1
    \end{pmatrix}\, .
\end{align}
Following this, act simultaneously on the $i+1$ qubit and the remaining qubits with a series of two-qubit gates given by 
\begin{align}
    T_{j,j+1}\coloneqq\frac{1}{\sqrt{i-j+1}}
    \begin{pmatrix} 
    \sqrt{i-j+1} & 0 & 0 & 0\\
    0 & 1 & \sqrt{i-j} & 0\\
    0 & -\sqrt{i-j} & 1 & 0\\
    0 & 0 & 0 & \sqrt{i-j+1}
    \end{pmatrix}\, ,
\end{align}
where $j$ ranges from 1 to $i-1$. This will give the desired control state. In all likelihood, there are even more ways to generate the desired control register. Whichever approach is chosen, the control state should remain the same. Note that the ancilla cost of the control state should be at least $\mathcal{O}(k \log_2 k)$ ancilla qubits regardless simply from the magnitude of the symmetric group, $|S_k|=k!$. 

Given this construction, it is easy to see the number of controlled-SWAP gates needed to perform the symmetric group test. Indeed, from Figure~\ref{fig:permutation4}, we see that the number of controlled-SWAPs needed when $k=4$ is $1+2+3=6$, where the 1 corresponds to the permutation $(1\ 2)$ needed to generate $S_2$, the 2 corresponds to the permutations $(2\ 3)$ and $(1\ 3)$ needed to generate $S_3$ from $S_2$, and the 3 corresponds to the permutations $(3\ 4)$, $(2\ 4)$, and $(1\ 4)$ needed to generate $S_4$ from $S_3$. By induction, the number of controlled-SWAP's needed to perform the $k$-th symmetric group test is $k(k-1)/2$, thus leading to the claimed $\mathcal{O}(k^2)$ gate complexity.

\begin{figure}[ptb]
\begin{center}
\includegraphics[
width=4.0in
]{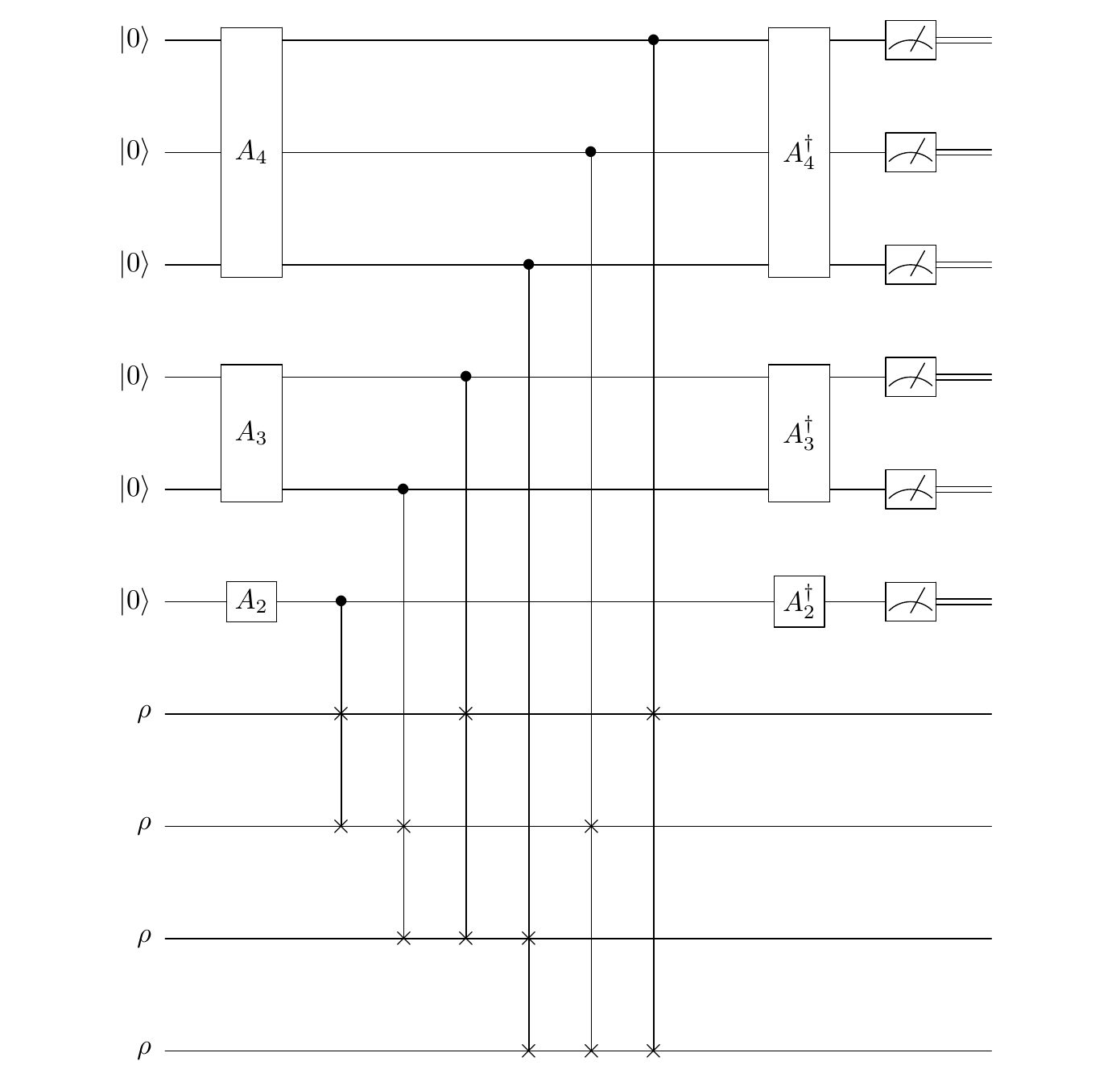}
\end{center}
\caption{Figure demonstrating how to systematically generate a test for the symmetric group of order four. }
\label{fig:permutation4}
\end{figure}

\subsubsection{Dihedral Group}
The dihedral group, $D_k$ is isomorphic to the semi-direct product of $\mathbb{Z}_k$ with $\mathbb{Z}_2$, with $\mathbb{Z}_2$ acting on $\mathbb{Z}_k$ by inversion. As such, it can be formed in a faithful way using a cyclic group generator and a non-commuting action that squares to identity. Using just the generators of the group, it is clear that the unitary flip action adds a factor of two to the number of cyclic gates needed, plus the additional instance of the flip element acting alone. In this manner, the full dihedral group requires at most $2k \log_2(k)$ gates to implement.

\subsection{Comparison between Subgroups of the Symmetric Group}

Now that we have given a method to count the number of quantum gates necessary for these separability tests, we consider if there is any advantage to using a simpler group as $k$ increases. Essentially, when is the trade-off between additional gates and acceptance probability favorable towards the various tests? 

The inherent motivation behind increasing $k$ is to obtain a smaller acceptance probability, prompting the need for Proposition~\ref{th:decreasing}. Clearly, the symmetric test provides the most stringent bound (see Figure~\ref{fig:acceptance}), yet it grows quickly in terms of gate resources needed (see Figure~\ref{fig:resource}). The cyclic group, however, benefits from the simplest construction but does not decay as quickly as the full symmetric group. It should be noted that, while it certainly appears to be true, we have not shown rigorously that the tests corresponding to sequences of groups other than the symmetric groups are decreasing. We suspect that a generating function argument might do the trick, but were not able to produce this result.

\begin{figure}[ptb]
\begin{center}
\includegraphics[
width=4.0in
]{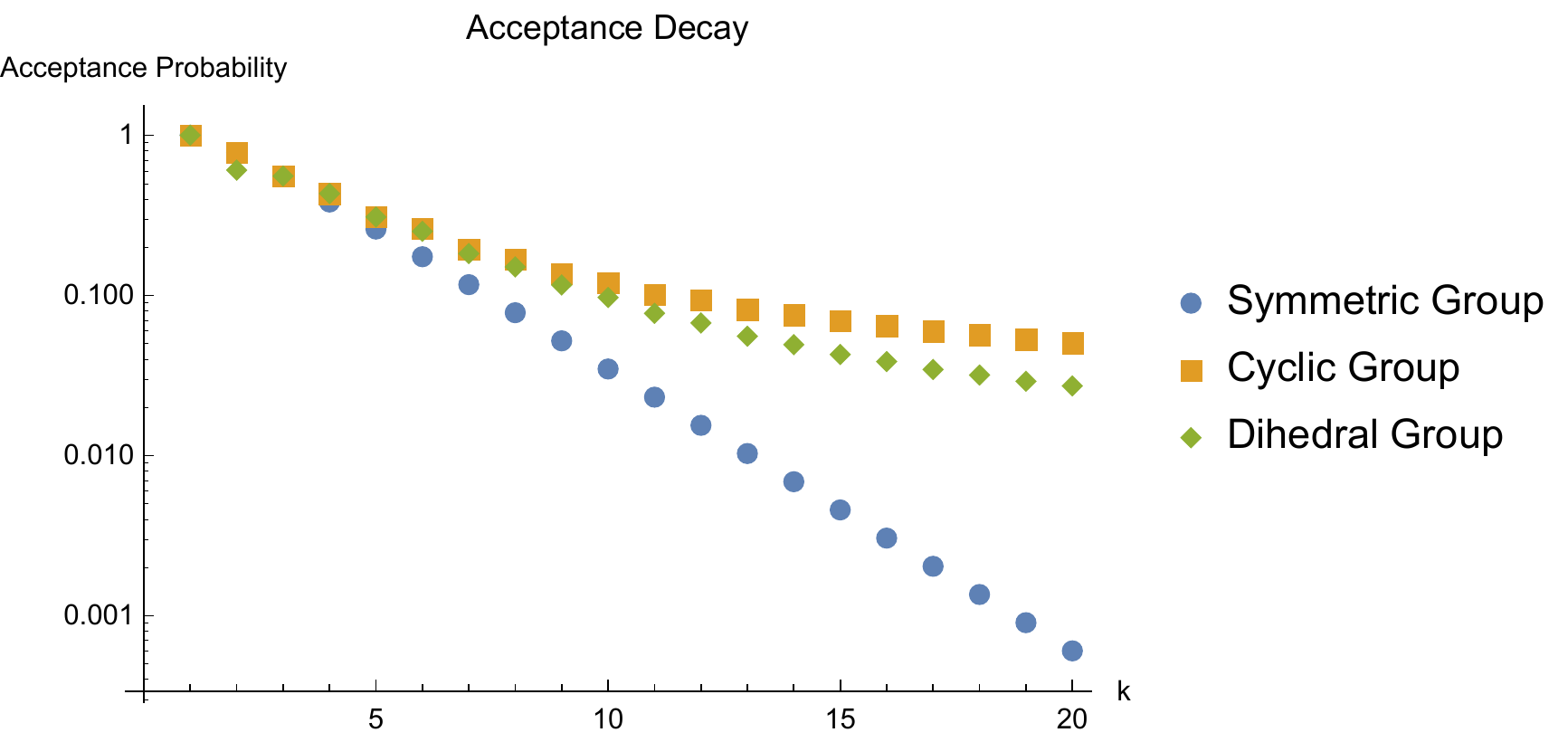}
\end{center}
\caption{Plot of the acceptance probabilities of each separability test as $k$ increases, for the symmetric group $S_k$, the cyclic group $C_k$, and the dihedral group $D_k$. For this example, we use a reduced $W$-state as an example to illustrate the algorithmic scaling for an unextendible state. For a separable state, all acceptance probabilities are equal to one.}
\label{fig:acceptance}%
\end{figure}

\begin{figure}[ptb]
\begin{center}
\includegraphics[
width=4.0in
]{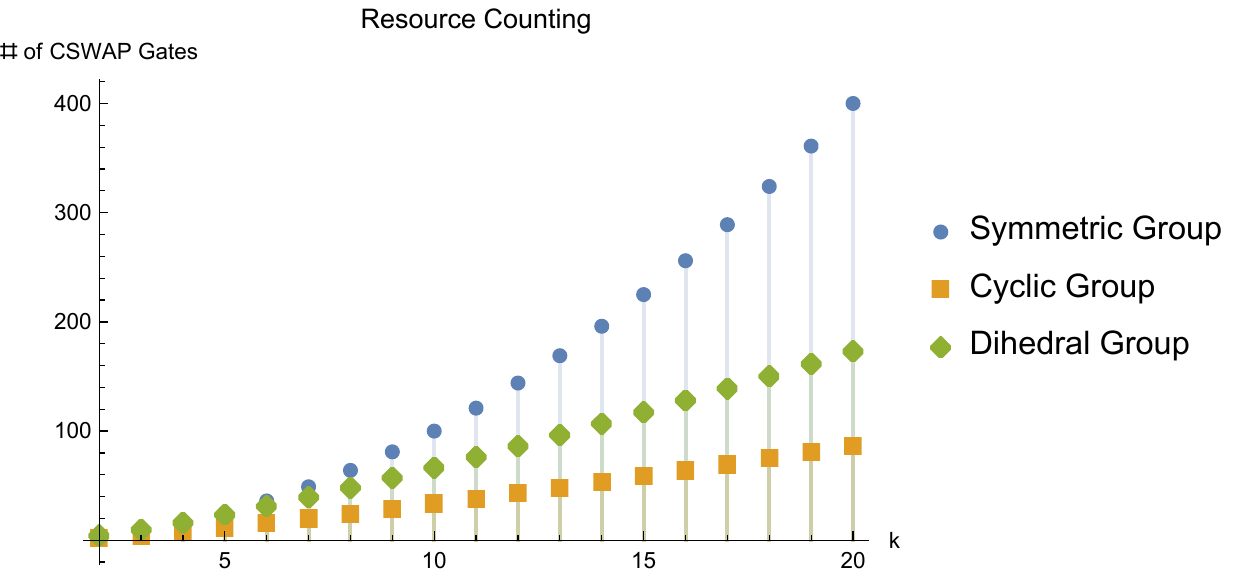}
\end{center}
\caption{Plot of the resource scaling in terms of the number of controlled-SWAP gates used for each group test as $k$ increases. We consider the symmetric group $S_k$, the cyclic group $C_k$, and the dihedral group $D_k$, and we use the gate counting methods described in the text.}
\label{fig:resource}%
\end{figure}

To visualize this trade-off, we consider the quantity $ R_{\textrm{test}}/(1-P_{\textrm{acc}})$, where $R_{\textrm{test}}$ is the number of resources needed to perform the test via the counting methods described above and $P_{\textrm{acc}}$ is the acceptance probability of the test. We employ the quantity $1-P_{\textrm{acc}}$ in the denominator, as we would like the test to have a lower acceptance probability for non-separable states, and thus the denominator will converge to one for better algorithms. This quantifier is very closely aligned with the expected runtime of the algorithm until getting a failure (it would be exactly equal to the expected runtime if we instead used circuit depth over $1-P_{\textrm{acc}}$ as the figure of merit). Thus, in comparing this quantity for the various tests, smaller values correspond to more ideal behavior from the algorithms. In Figure~\ref{fig:ratio}, we show this quantity for algorithms generated by the cyclic group and the symmetric group.

\begin{figure}[ptb]
\begin{center}
\includegraphics[
width=4.0in
]{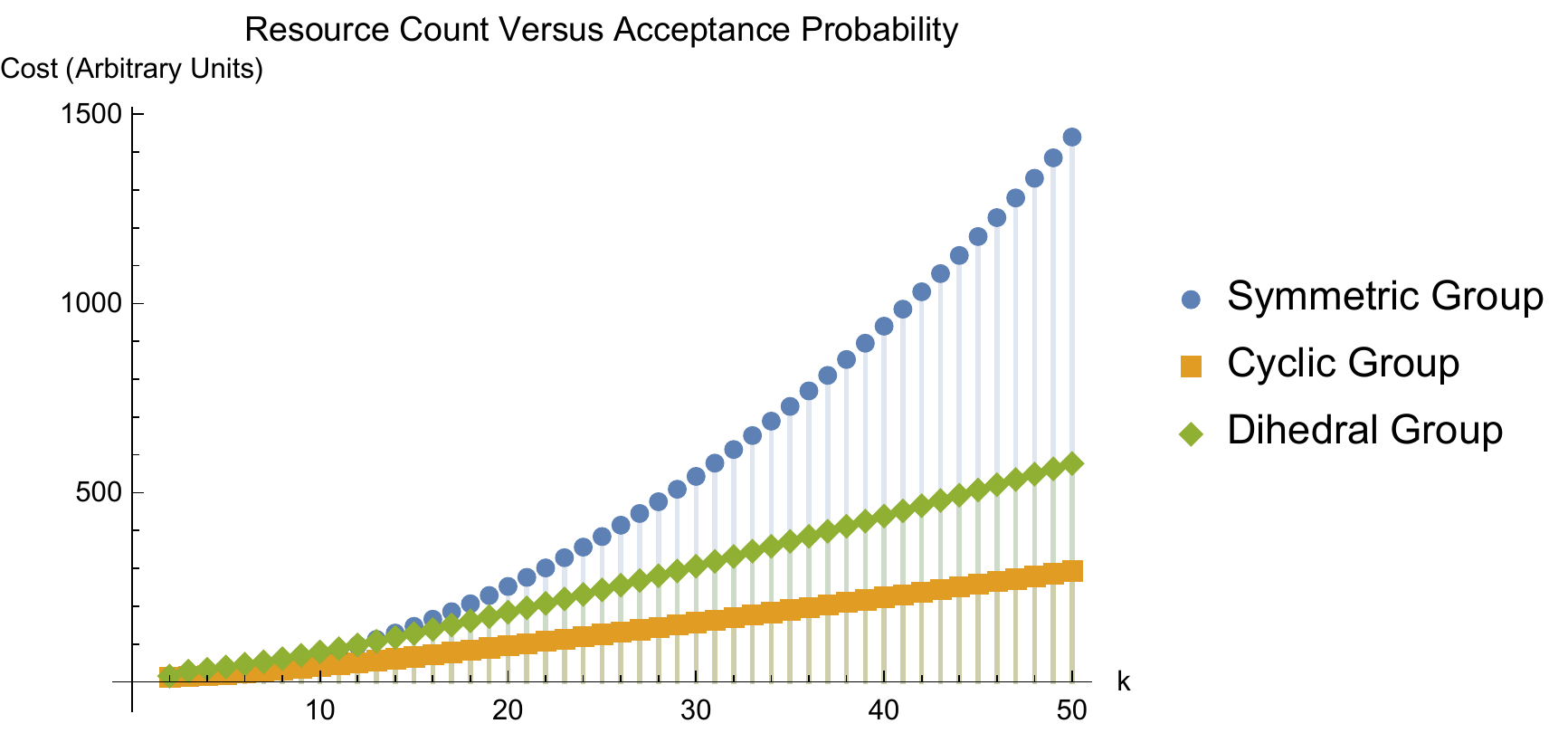}
\end{center}
\caption{We show the ratio of the resources required to rejection probability as $k$ increases. We consider here the cyclic group of $k$ elements as an example of a simple Abelian group and show it gives an advantage in terms of the resources-to-rejection metric over the test generated by the full symmetric group.}
\label{fig:ratio}
\end{figure}

Examining Figure~\ref{fig:ratio}, we see a clear difference in the performance between the tests generated by the cyclic group and the symmetric group. The plotted ratio can be thought of as resources-to-rejection, in the sense that $1-P_{\textrm{acc}}$ is the probability that a non-separable state is correctly identified---or rather, the failure rate of the algorithm for such a state. Although Figure~\ref{fig:acceptance} makes it appear that the standard test generated from $S_k$ would always be preferable, we determine from this comparison that the $C_k$ algorithm gives more benefit per gate resource. 

From this analysis, we can assert that the simpler test for separability is more cost-efficient than the full permutation test. We show in Figure~\ref{fig:acceptance} that both tests show a decrease in acceptance probability as $k$ increases, a desirable trait. However, Figure~\ref{fig:resource} shows how quickly circuit sizes grow as $k$ increases, particularly for $S_k$, which can be considered the standard test. Figure~\ref{fig:ratio} bridges these notions to show that the comparative growth in gate resources of $S_k$ outweighs the relative decrease in acceptance probability given over the $C_k$ test. We thus determine that the cyclic group $C_k$ suffices as a separability test of this nature.

\section{Conclusion}

\label{sec:conclusion}

In this work, we have presented several separability tests for bipartite pure states, and we have established analytical expressions for their acceptance probabilities. These expressions invariably rely on the cycle index polynomial of the group. Indeed, from a mathematical point of view, this relationship seems natural, due to the inherent combinatorics present in the algorithms. Nonetheless, these expressions give us direct insight into the performance of any separability tests generated from a finite group---which we have shown can be feasibly constructed. Using this perspective, we demonstrate that when utilizing more copies of the state under test, these tests become more stringent. Additionally, we observe that the full symmetric test using a representation of the symmetric group gives a quickly decreasing acceptance probability for an entangled state; however, for the given implementations of these algorithms, other tests can use fewer resources and still show great efficiency. 

Here, we have limited ourselves to pure bipartite states; however, we believe multipartite tests may yield results in a similar vein. For instance, a trivial implementation would be to separate all parties into individual tests and then multiply the results. There is a question, however, if more elegant algorithms exist for multipartite cases, and if interesting mathematics arise in the study of such systems. Additionally, generalizing these algorithms to apply to mixed states is a natural next step. In particular, providing tests for separability of multipartite mixed states would be an interesting subject for further study. Furthermore, we think it would be interesting to study the performance of these algorithms on near-term quantum devices; all of our analysis in this paper is for the ideal case.

For future work, we propose considering further the nature of the relationship between cycle index polynomials and quantum algorithms. Our results have shown a clear benefit in the use of cycle index polynomials in describing separability tests; is the reverse true? An interesting question is whether algorithms of this form can be used to determine the coefficients of cycle index polynomials in general, and we leave this for future work.

\section*{Acknowledgements}
MLL acknowledges support from the DOD Smart Scholarship program, and MMW from the National Science Foundation under Grant No.~1907615. ZPB acknowledges Christophe Vignat for his helpful comments on a draft of this work. MLL would like to acknowledge Kate Caillet for helpful discussions.

\section*{Data Availability Statement}
The datasets generated during and/or analysed during the current study are available in the GitHub repository, \sloppy \url{https://github.com/mlabo15/GeneralizedSeparability}.

\section*{Competing Interests Statement}
MLL acknowledges support from the DOD Smart Scholarship program, and MMW from the National Science Foundation under Grant No.~1907615. The authors have no non-financial competing interests to declare.

\nocite{nielsenchuang,wildebook}

\bibliographystyle{unsrt}
\bibliography{Ref}

\end{document}